\documentclass[11pt,english,aps,manuscript]{revtex4}
\usepackage{libertine-type1}
\usepackage[T1]{fontenc}
\usepackage[latin9]{inputenc}
\usepackage{amsthm}
\usepackage{amsmath}
\usepackage{amssymb}

\makeatletter
\@ifundefined{textcolor}{}
{%
 \definecolor{BLACK}{gray}{0}
 \definecolor{WHITE}{gray}{1}
 \definecolor{RED}{rgb}{1,0,0}
 \definecolor{GREEN}{rgb}{0,1,0}
 \definecolor{BLUE}{rgb}{0,0,1}
 \definecolor{CYAN}{cmyk}{1,0,0,0}
 \definecolor{MAGENTA}{cmyk}{0,1,0,0}
 \definecolor{YELLOW}{cmyk}{0,0,1,0}
}
  \theoremstyle{definition}
  \newtheorem{defn}{\protect\definitionname}
  \theoremstyle{plain}
  \newtheorem{prop}{\protect\propositionname}
 \theoremstyle{definition}
  \newtheorem{example}{\protect\examplename}
  \theoremstyle{plain}
  \newtheorem{cor}{\protect\corollaryname}
\theoremstyle{plain}
\newtheorem{thm}{\protect\theoremname}
  \theoremstyle{plain}
  \newtheorem{lem}{\protect\lemmaname}

\usepackage{babel}

\usepackage{babel}
\providecommand{\definitionname}{Definition}
  \providecommand{\examplename}{Example}
  \providecommand{\lemmaname}{Lemma}
  \providecommand{\propositionname}{Proposition}
\providecommand{\corollaryname}{Corollary}
\providecommand{\theoremname}{Theorem}

\usepackage{babel}
\providecommand{\definitionname}{Definition}
  \providecommand{\examplename}{Example}
  \providecommand{\lemmaname}{Lemma}
  \providecommand{\propositionname}{Proposition}
\providecommand{\corollaryname}{Corollary}
\providecommand{\theoremname}{Theorem}

\usepackage{babel}
\providecommand{\definitionname}{Definition}
  \providecommand{\examplename}{Example}
  \providecommand{\lemmaname}{Lemma}
  \providecommand{\propositionname}{Proposition}
\providecommand{\corollaryname}{Corollary}
\providecommand{\theoremname}{Theorem}

\makeatother

\usepackage{babel}
  \providecommand{\definitionname}{Definition}
  \providecommand{\examplename}{Example}
  \providecommand{\lemmaname}{Lemma}
  \providecommand{\propositionname}{Proposition}
\providecommand{\corollaryname}{Corollary}
\providecommand{\theoremname}{Theorem}

\begin{document}

\title{Optimal resource states for local state discrimination }

\author{Somshubhro Bandyopadhyay}

\email{som@jcbose.ac.in}

\selectlanguage{english}%

\author{Saronath Halder }

\affiliation{Department of Physics and Center for Astroparticle Physics and Space
Science, Bose Institute, EN 80, Sector V, Bidhannagar, Kolkata 700091,
India}

\author{Michael Nathanson}

\email{man6@stmarys-ca.edu }

\selectlanguage{english}%

\affiliation{Department of Mathematics and Computer Science, Saint Mary's College
of California, Moraga, CA, 94556, USA}
\begin{abstract}
We study the problem of locally distinguishing pure quantum states
using shared entanglement as a resource. For a given set of locally
indistinguishable states we define a resource state to be \textit{useful}
if it can enhance local distinguishability and \textit{optimal} if
it can distinguish the states as well as global measurements and is
also minimal with respect to a partial ordering defined by entanglement
and dimension. We present examples of useful resources and show that
an entangled state need not be useful for distinguishing a given set
of states. We obtain optimal resources with explicit local protocols
to distinguish multipartite $GHZ$ and Graph states; and also show
that a maximally entangled state is an optimal resource under one-way
LOCC to distinguish any bipartite orthonormal basis which contains
at least one entangled state of full Schmidt rank.  
\end{abstract}
\maketitle

\section{Introduction}

The paradigm of Local Operations and Classical Communication (LOCC)
\cite{Chitambar-LOCC} is of central importance in quantum information
theory. In a LOCC protocol, two or more distant parties perform arbitrary
quantum operations on local subsystems and communicate classically,
but are not allowed to exchange quantum information (qubits). Fundamental
questions on quantum nonlocality, and properties of entangled states
(see \cite{Entanglement-horodecki} for a review), especially those
related to the notion of entanglement as a resource, are generally
explored within the framework of LOCC.

LOCC protocols have limitations in that they cannot implement all
quantum operations on a composite system, parts of which are spatially
separated. For example, it is impossible, by LOCC, to entangle two
or more quantum systems, even with nonzero probability. Shared entanglement,
however, can help to overcome such limitations. The local protocols
which use shared entanglement as a resource define the class of LOCCE,
short for Local Operations, Classical Communication and Entanglement.
These protocols, using appropriate entangled states, can enable local
implementation of any quantum operation on the whole system. It is
in this sense, we say that entanglement is a resource for quantum
operations, e.g. quantum teleportation \cite{Teleportation}, superdense
coding \cite{Densecoding}, entanglement catalysis \cite{Jonathan-Plenio-1999},
entangling measurements and unitaries \cite{Berry-2007,Cirac-et-al-2001,Collins-et-al-2001,BBKW-2009,BRW-2010}.
The present paper considers a problem along these lines, namely quantum
state discrimination by LOCCE. This problem has been previously explored,
primarily in specific instances of bipartite systems \cite{Collins-et-al-2001,BBKW-2009,Bennett-I-99,Cohen-2008,B-IQC-2015,Yu-Duan-Ying,Nathanson-2013},
while a recent work \cite{BHN-2016} initiated a more general treatment
of both bipartite and multipartite systems.

In a local state discrimination problem \cite{Bennett-I-99,Yu-Duan-Ying,Peres-Wootters-1991,Walgate-2000,Virmani-2001,Ghosh-2001,Walgate-2002,Ghosh-2002,Bennett-II-99 +Divin-2003,HSSH2003,Ghosh-2004,Watrous-2005,Nathanson-2005,Duan2007,Duan-2009,Calsamiglia-2010,BGK,Bandyo-2011,Yu-Duan-2012,BN-2013,Cosentino-2013,Cosentino-Russo-2014},
the goal is to learn about the state of a multipartite quantum system,
prepared in one of a known set of states, by LOCC measurements. In
some cases, LOCC can indeed perform this task optimally, i.e. as well
as global measurements. For example, any two pure states can be optimally
distinguished by LOCC regardless of their dimensions, entanglement
and multipartite structure \cite{Walgate-2000,Virmani-2001}. On the
other hand, there exist states which cannot be optimally distinguished
by LOCC, and such states are said to be locally indistinguishable
(LI); e.g. three Bell states \cite{Ghosh-2001}, a complete orthogonal
basis where not all states are product \cite{Ghosh-2001,Ghosh-2002,HSSH2003,Ghosh-2004,Nathanson-2005},
the orthogonal product bases exhibiting ``nonlocality without entanglement''
\cite{Bennett-I-99}, and unextendible product bases \cite{Bennett-II-99 +Divin-2003}.
LI states are said to exhibit a new kind of nonlocality as emphasized
by many authors \cite{Bennett-I-99,Bennett-II-99 +Divin-2003,HSSH2003,Bandyo-2011}
and imply that global information encoded in multipartite systems
may not be completely accessible by local means \cite{Ghosh-accessible information,Horodecki-accessible-information}.
The latter has found useful applications in data hiding \cite{Terhal2001,DiVincenzo2002,Eggeling2002,MatthewsWehnerWinter}
and secret sharing \cite{Markham-Sanders-2008}.

The existence of locally indistinguishable (LI) states imply that
auxiliary entanglement, shared between the parties, may be necessary
for optimal discrimination of such states by LOCC. Indeed, entanglement
is necessary to perfectly distinguish any bipartite or multipartite
orthonormal basis containing entangled states \cite{BBKW-2009,Ghosh-2001,Ghosh-2002,HSSH2003}.
On the other hand, with sufficiently many entangled states any set
of LI states can be optimally distinguished. For example, the teleportation
protocol \cite{Collins-et-al-2001,Cohen-2008,BHN-2016} can optimally
distinguish any set of LI states in $\left(\mathbb{C}^{d}\right)^{\otimes N}$
while consuming $\left(N-1\right)\log d$ $ebits$. However, from
a resource perspective the teleportation protocol in general is not
optimal. For example, Cohen \cite{Cohen-2008} presented protocols
which use entanglement more efficiently than teleportation to perfectly
distinguish certain classes of unextendible product bases.

The purpose of the present work is to better understand the role of
entanglement, as a resource, in local state discrimination problems.
We therefore focus on the characterization of resource states and
also present results on multipartite state discrimination that are
optimal under LOCCE.

Local fidelity \cite{Navascues}, which quantifies how well a set
of states can be distinguished by LOCC, plays a central role in our
analysis. This is briefly reviewed in Section II vis-a-vis the problem
of local state discrimination. In Section III we give a sufficiently
general formulation of the problem of state discrimination under LOCCE.
Although we concern ourselves only with distinguishing pure states,
the formulation and much of the subsequent analysis can be easily
extended to mixed states.

In Section IV we define the resource states to be useful or optimal
for a given set $\mathcal{S}$ of locally indistinguishable states
and illustrate these definitions with some general results and examples.
We say that a resource state is useful if and only if it can enhance
local distinguishability of the states in $\mathcal{S}$. In bipartite
systems, we show that any pure entangled state of Schmidt rank $r\geq2$
is useful for distinguishing the elements of any $\mathcal{S}\subset\mathbb{C}^{d_{1}}\otimes\mathbb{C}^{d_{2}}$,
$2\leq d_{1}\leq d_{2}$ provided $r\geq d_{1}$. We also show that
a resource state can be useful when $2\leq r<d_{1}$: in particular,
$m$ copies of a Bell state are shown useful to distinguish an orthonormal
basis of Lattice states in $\left(\mathbb{C}^{2}\right)^{\otimes n}\otimes\left(\mathbb{C}^{2}\right)^{\otimes n}$
for any $1\leq m\leq n$.

Here one is tempted to ask: For a given set $\mathcal{S}$, is every
pure entangled state useful as a resource? We answer this question
in negative. As an example, we prove that any pure state with only
bipartite entanglement is not useful for distinguishing a three-qubit
$GHZ$ basis. Similar arguments show that any $N^{\prime}$-partite
state, no matter how entangled, cannot be useful for distinguishing
a $N$-qubit $GHZ$ basis whenever $N^{\prime}\leq N-1$.

Next we consider the question of optimality of resource states. Let
$\mathcal{R}$ be the set of all resource states that optimally distinguishes
the states in $\mathcal{S}$ under LOCCE. Since most states in $\mathcal{R}$
are not optimal from a resource point of view, we give two additional
conditions that an optimal resource $\left|\Psi\right\rangle \in\mathcal{R}$
must satisfy. The first condition is that the amount of entanglement
consumed is no more than what is both necessary and sufficient. This
can be duly satisfied by requiring that $E\left(\Psi\right)\leq E\left(\Psi^{\prime}\right)$
for any $\left|\Psi^{\prime}\right\rangle \in\mathcal{R}$ where $E$
is a well-defined measure of entanglement. The second condition requires
that the dimension of the optimal resource space must be the smallest;
i.e. $\dim\mathcal{H}_{\Psi}\leq\dim\mathcal{H}_{\Psi^{\prime\prime}}$
for any $\left|\Psi^{\prime\prime}\right\rangle \in\mathcal{R}$ satisfying
$E\left(\Psi^{\prime\prime}\right)=E\left(\Psi\right)$.

As examples, we obtain optimal resources for distinguishing $GHZ$
and Graph states in Section V. We show that a $m$-qubit $m$-partite
$GHZ$ state is an optimal resource for distinguishing a $N$-qubit
$m$-partite orthonormal $GHZ$ basis with $N\geq m\geq2$ for any
partitioning of the $N$ qubits among $m$ parties. This result is
generalized to Graph states where an optimal resource is obtained
for distinguishing a basis defined for any graph $G$ on $N$ vertices
with each party holding a qubit. Section VI considers optimal resources
in one-way LOCCE for bipartite systems. Here we show that a maximally
entangled state is an optimal resource for distinguishing any bipartite
orthonormal basis containing an entangled state of full Schmidt rank.
We conclude in Section VII with a discussion on some of the open problems.

\section{Preliminaries: Quantum state discrimination and Fidelity}

In a quantum state discrimination problem, we wish to quantify how
much can be learned about a quantum system, prepared in one of a known
set $\mathcal{S}=\left\{ p_{i},|\psi_{i}\rangle\right\} $ of pure
quantum states $|\psi_{i}\rangle$ occurring with probabilities $p_{i}$.
The average fidelity is one measure calculated with respect to a particular
physical protocol and a decoding scheme, defined initially in \cite{Fuchs-Sasaki}.
Thus, for fixed set $\mathcal{S}=\left\{ p_{i},|\psi_{i}\rangle\right\} $,
a measurement (POVM) $\mathbf{M}=\left\{ M_{a}\right\} $, and a guessing
strategy $\mathbf{G}:a\rightarrow|\phi_{a}\rangle$, the average fidelity
is given by \cite{Fuchs-Sasaki,Navascues}, 
\begin{eqnarray}
\mathbb{F}\left(\mathcal{S}\vert\mathbf{M},\mathbf{G}\right) & = & \sum_{i,a}p_{i}\left\langle \psi_{i}\left|M_{a}\right|\psi_{i}\right\rangle \left|\left\langle \psi_{i}|\phi_{a}\right\rangle \right|^{2}\label{avg-fidelity}
\end{eqnarray}
This measures our ability to prepare a new quantum system in a state
which is close to the original state $\vert\psi_{i}\rangle$. It may
be noted that $0\leq\mathbb{F}\left(\mathcal{S}\vert\mathbf{M},{\bf G}\right)\leq1$,
and $\mathbb{F}\left(\mathcal{S}\vert\mathbf{M},\mathbf{G}\right)=1$
if and only if the procedure $\left(\mathbf{M},\mathbf{G}\right)$
identifies the given state of our system perfectly which is possible
only if the states $\vert\psi_{i}\rangle$ are orthogonal. The optimal
fidelity is defined as 
\begin{equation}
\mathbb{F}_{{\rm opt}}\left(\mathcal{S}\right)=\sup_{\mathbf{M}\in{\rm ALL},\mathbf{G}}\mathbb{F}\left(\mathcal{S}\vert\mathbf{M},\mathbf{G}\right),\label{F-global}
\end{equation}
where the optimization is over all quantum measurements and all guessing
strategies. Note that the problem of finding an optimal measurement
strategy is difficult in general but that for any fixed measurement
$\mathbf{M}$, it is a straightforward calculation to calculate the
optimal guessing strategy $\mathbf{G}$.

%\subsection{Quantum state discrimination by LOCC and Local Fidelity}

In LOCC state discrimination we suppose that the states $\vert\psi_{i}\rangle\in\mathcal{S}$
belong to a $N$-partite quantum system $\mathcal{H}_{{\rm S}}=\otimes_{i=1}^{N}\mathbb{C}^{d_{i}}$,
$N\geq2$ with the allowed measurements belonging to the LOCC class.
If we label the parties as $A_{1},A_{2},\dots,A_{N}$, then the LOCC
measurements are realized with respect to the partitioning $A_{1}\vert A_{2}\vert\cdots\vert A_{N}$
unless stated otherwise. To simplify the notation, a LOCC protocol
$\mathbf{P}$ will denote the associated LOCC measurement $\mathbf{M}$
and the corresponding guessing strategy $\mathbf{G}$. Accordingly,
the optimal local fidelity $\mathbb{F}_{{\rm local}}\left(\mathcal{S}\right)$
is defined as \cite{Navascues} 
\begin{equation}
\mathbb{F}_{{\rm local}}\left(\mathcal{S}\right)=\sup_{\mathbf{P}\in{\rm LOCC}}\mathbb{F}\left(\mathcal{S}\vert\mathbf{P}\right)\leq\mathbb{F}_{{\rm opt}}\left(\mathcal{S}\right),\label{F-local}
\end{equation}
where the optimization is over all LOCC protocols. We say that the
states $\left|\psi_{i}\right\rangle $ are locally indistinguishable
if and only if $\mathbb{F}_{{\rm local}}\left(\mathcal{S}\right)<\mathbb{F}_{{\rm opt}}\left(\mathcal{S}\right)$.

In multipartite systems with $N\geq3$ we can often determine whether
a given set of states is locally indistinguishable or not by examining
local distinguishability across various bipartitions. Let $A\vert B$
be a bipartition where $A$ and $B$ hold $m$ and $\left(N-m\right)$
subsystems respectively. The optimal local fidelity across this bipartition
is given by 
\begin{equation}
\mathbb{F}_{{\rm local}}\left(\mathcal{S}_{A\vert B}\right)=\sup_{\mathbf{P}\in{\rm LOCC}}\mathbb{F}\left(\mathcal{S}_{A\vert B}\vert\mathbf{P}\right),\label{bipartition local F}
\end{equation}
where the optimization is now over all LOCC protocols $\mathbf{P}$
with respect to the bipartition $A\vert B$. Since local fidelity
cannot increase by further partitioning of the subsystems, the following
inequality holds: 
\begin{eqnarray}
\mathbb{F}_{{\rm local}}\left(\mathcal{S}\right) & \leq & \min_{\left\{ A\vert B\right\} }\mathbb{F}_{{\rm local}}\left(\mathcal{S}_{A\vert B}\right),\label{Flocal<F-bipartition}
\end{eqnarray}
where the minimum is obtained over all bipartitions. It is clear from
the above inequality that the states must be locally indistinguishable
if they are locally indistinguishable across at least one bipartition.
However, it should be noted that a set of states can be locally indistinguishable
even though they can be optimally distinguished by LOCC across every
bipartition \cite{Bennett-II-99 +Divin-2003}.

\section{Quantum state discrimination by LOCCE}

We now suppose that the states $\vert\psi_{i}\rangle\in\mathcal{S}$
are locally indistinguishable, and furthermore that the parties share
a resource state $\vert\Psi\rangle$ in addition to an unknown element
of $\mathcal{S}$. Since in multipartite systems it may not be always
necessary that the resource state is shared by all, we suppose that
$\left|\Psi\right\rangle $ belongs to a $N^{\prime}$-partite quantum
system $\mathcal{H}_{{\rm \Psi}}=\otimes_{i=1}^{N^{\prime}}\mathbb{C}^{d_{i}^{\prime}}$,
where $2\leq N^{\prime}\leq N$. Note that we do not make any assumption
on the structure of the resource state; it can be either genuine multipartite
entangled, e.g. a $GHZ$ state or a tensor product of entangled states.

Observe that the task of local discrimination of the states $\left|\psi_{i}\right\rangle $
using a resource state $\vert\Psi\rangle$ is, in fact, equivalent
to the task of local discrimination of the states $\vert\Psi\rangle\otimes\vert\psi_{i}\rangle$
where the states $\vert\Psi\rangle\otimes\vert\psi_{i}\rangle$ now
belong to an enlarged joint Hilbert space $\mathcal{H}=\mathcal{H}_{\Psi}\otimes\mathcal{H}_{{\rm S}}$.
In this setting, LOCC is understood as follows: Relabel the parties
who share the resource state as $A_{1},A_{2},\dots,A_{N^{\prime}}$,
$2\leq N^{\prime}\leq N$. With this, the joint Hilbert space can
be expressed as $\mathcal{H}=\otimes_{i=1}^{N}\mathcal{H}_{i}$ where
$\mathcal{H}_{i}=\mathbb{C}^{d_{i}^{\prime}}\otimes\mathbb{C}^{d_{i}}$
for $i=1,\dots,N^{\prime}$ and $\mathcal{H}_{i}=\mathbb{C}^{d_{i}}$
for $i=(N^{\prime}+1),\dots,N$. This means that for $i=1,\dots,N^{\prime}$
each party $A_{i}$ holds two quantum systems: the first system of
dimension $d_{i}^{\prime}$ is part of the shared resource state,
and the second system, of dimension $d_{i}$ is part of the shared
unknown state, and therefore, \emph{joint} quantum operations are
allowed on these two systems.

As before, local distinguishability of the states $\vert\Psi\rangle\otimes\vert\psi_{i}\rangle$
can be duly quantified by the local fidelity $\mathbb{F}\left(\Psi\otimes\mathcal{S}\vert\mathbf{P}\right)$
for some LOCC protocol $\mathbf{P}$. The optimal local fidelity is
defined in (\ref{F-local}) as 
\begin{equation}
\mathbb{F}_{{\rm local}}\left(\Psi\otimes\mathcal{S}\right)=\sup_{\mathbf{P}\in{\rm LOCC}}\mathbb{F}\left(\Psi\otimes\mathcal{S}\vert\mathbf{P}\right)\leq\mathbb{F}_{{\rm opt}}\left(\mathcal{S}\right).\label{F-local-1}
\end{equation}
Computing the local fidelity is sufficient to ensure how well a given
set of states can be locally distinguished using a given resource
state. We use this fact in the next section for our definition of
useful resources. On the other hand, local fidelity doesn't tell us
anything about \textquotedbl{}efficient\textquotedbl{} use of entanglement
in the process of local state discrimination. As one may recall, it
is easy to achieve the global optimum via the teleportation protocol
with $\left|\Psi\right\rangle $ being a tensor product of many bipartite
maximally entangled states. In the next section, we therefore lay
down the criteria an optimal resource must satisfy besides achieving
the global optimum.

\section{Characterization of the resource states: Useful and Optimal resources}

In this section we discuss useful and optimal resources in a state
discrimination problem under LOCCE.

\subsection{Useful resource states}

In quantum information theory, an entanglement state is considered
useful for a task if and only if it helps to perform the task better
than LOCC alone. In the same spirit we define useful resources in
local state discrimination. 
\begin{defn}
\label{useful-def} For a given set $\mathcal{S}$ of locally indistinguishable
states, a resource state $\vert\Psi\rangle$ is \emph{useful} iff
there exists a LOCC protocol $\mathbf{P}$ %(not necessarily optimal LOCC) suchsuchsuch
that $\mathbb{F}\left(\Psi\otimes\mathcal{S}\vert\mathbf{P}\right)>\mathbb{F}_{{\rm local}}\left(\mathcal{S}\right)$.
That is: $\mathbb{F}_{{\rm local}}\left(\Psi\otimes\mathcal{S}\right)>\mathbb{F}_{{\rm local}}\left(\mathcal{S}\right)$ 
\end{defn}
Thus, $\vert\Psi\rangle$ is useful iff it can enhance the distinguishability
of the set $\mathcal{S}$ under LOCC. For a fixed $\mathcal{S}$ it
seems difficult to ascertain whether a given resource state is useful
or not, but nonetheless, we present some general results and examples
that answer some of the closely related questions.

Consider a set $\mathcal{S}$ of LI states in $\mathbb{C}^{d_{1}}\otimes\mathbb{C}^{d_{2}}$,
$d_{1}\leq d_{2}$. Proposition \ref{useful-prop} shows that any
pure entangled state of Schmidt rank $r$ is useful if $r\geq d_{1}$.
We also show that useful resource states with $r<d_{1}$ exist. In
Example \ref{Lattice State Example} a resource state of the form
$\left|\Phi\right\rangle ^{\otimes m}$ where $\left|\Phi\right\rangle $
is a Bell state, is shown to be useful in distinguishing an orthonormal
basis of Lattice states in $\left(\mathbb{C}^{2}\right)^{\otimes n}\otimes\left(\mathbb{C}^{2}\right)^{\otimes n}$
for any $1\leq m\leq n$.

In Example \ref{GHZ-example} we show that any bipartite pure entangled
state is not a useful resource for distinguishing a three-qubit $GHZ$
basis. Generalizing this, we point out that a $N^{\prime}$-partite
pure entangled state, no matter how entangled, cannot be considered
useful for distinguishing a $N$-qubit $GHZ$ basis provided $N^{\prime}\leq N-1$. 
\begin{prop}
\label{useful-prop} Let $\mathcal{S}$ be a set of locally indistinguishable
states in $\mathcal{H}_{{\rm S}}=\mathbb{C}^{d_{1}}\otimes\mathbb{C}^{d_{2}}$,
$2\leq d_{1}\leq d_{2}$. Then any bipartite pure state $\vert\Psi\rangle\in\mathcal{H}_{\Psi}$
of Schmidt rank $r\geq d_{1}$ is a useful resource. \end{prop}
\begin{proof}
We prove the proposition by giving an explicit local protocol. Since
$r\geq d_{1}$ we assume that $\left|\Psi\right\rangle \in\mathbb{C}^{d_{1}^{\prime}}\otimes\mathbb{C}^{d_{2}^{\prime}}$,
where $d_{1}\leq d_{1}^{\prime}\leq d_{2}^{\prime}$. In the first
step of the protocol we attempt to convert $\vert\Psi\rangle$ to
a maximally entangled state $\vert\Psi^{\prime}\rangle$ of Schmidt
rank $d_{1}$ by LOCC. From Vidal's theorem \cite{Vidal-1999} we
know that this local conversion succeeds with probability $p>0$,
and when it does, we use $\vert\Psi^{\prime}\rangle$ to optimally
distinguish the states by the teleportation protocol. The local conversion,
however, fails with probability $\left(1-p\right)$ in which case
we resort to the optimal LOCC measurement to distinguish the states.
Thus with respect to this local protocol, say $\mathbf{P}$, the fidelity
is given by 
\begin{eqnarray*}
\mathbb{F}\left(\Psi\otimes\mathcal{S}\vert\mathbf{P}\right) & = & p\mathbb{F}_{{\rm opt}}\left(\mathcal{S}\right)+\left(1-p\right)\mathbb{F}_{{\rm local}}\left(\mathcal{S}\right)\\
 & > & \mathbb{F}_{{\rm local}}\left(\mathcal{S}\right)
\end{eqnarray*}
since $p>0$ and $\mathbb{F}_{{\rm opt}}\left(\mathcal{S}\right)>\mathbb{F}_{{\rm local}}\left(\mathcal{S}\right)$.
Hence, $\left|\Psi\right\rangle $ is useful for locally distinguishing
$\mathcal{S}$. \end{proof}
\begin{example}
\label{Lattice State Example} Let $\vert\Phi_{i}\rangle$, $i=1,\dots,4$
be the states in the Bell basis $\mathcal{B}\subset\mathbb{C}^{2}\otimes\mathbb{C}^{2}$
where $\vert\Phi_{1}\rangle=\frac{1}{\sqrt{2}}\left(\vert00\rangle+\vert11\rangle\right)$,
$\vert\Phi_{2}\rangle=\frac{1}{\sqrt{2}}\left(\vert00\rangle-\vert11\rangle\right)$$,\vert\Phi_{3}\rangle=\frac{1}{\sqrt{2}}\left(\vert01\rangle+\vert10\rangle\right)$
and $\vert\Phi_{4}\rangle=\frac{1}{\sqrt{2}}\left(\vert01\rangle-\vert10\rangle\right)$.
Now consider a local state discrimination problem where two parties
hold $n\geq2$ unknown Bell states. This means that the unknown state
belongs to the maximally entangled basis 
\begin{eqnarray*}
\mathcal{B}^{\prime}=\left\{ \left.\bigotimes_{j=1}^{n}\vert\Phi_{i_{j}}\rangle\right\vert i_{j}\in\{1,2,3,4\}\right\} \subset\left(\mathbb{C}^{2}\right)^{\otimes n}\otimes\left(\mathbb{C}^{2}\right)^{\otimes n}
\end{eqnarray*}
The states in $\mathcal{B}^{\prime}$ are known as Lattice States
in the literature (see for example, \cite{Cosentino-2013}). It was
shown in \cite{Cosentino-2013,Cosentino-Russo-2014} that for $n\geq3$,
there are small subsets $\mathcal{S}\subset\mathcal{B}^{\prime}$
with $\vert\mathcal{S}\vert<2^{n}$ such that $\mathcal{S}$ is locally
indistinguishable. Here we restrict our attention to distinguishing
a complete basis.

Assuming that the states in $\mathcal{B}^{\prime}$ are all equiprobable,
we will show that the resource state $\vert\Psi_{m}\rangle=\vert\Phi_{1}\rangle^{\otimes m}$
is useful for any $1\leq m\leq n$, i.e. 
\begin{eqnarray*}
\mathbb{F}_{{\rm local}}\left(\Psi_{m}\otimes\mathcal{B}^{\prime}\right) & > & \mathbb{F}_{{\rm local}}\left(\mathcal{B}^{\prime}\right).
\end{eqnarray*}
Note that, $\Psi_{m}\otimes\mathcal{B}^{\prime}$ is a set of $2^{2n}$
maximally entangled states in $\left(\mathbb{C}^{2}\right)^{\otimes(n+m)}\otimes\left(\mathbb{C}^{2}\right)^{\otimes(n+m)}$.

We proceed as follows. We first bound $\mathbb{F}_{{\rm local}}\left(\mathcal{B}^{\prime}\right)$
using a result from \cite{BN-2013} which states that the local fidelity
when distinguishing a set of $k$ $\mathbb{C}^{d}\otimes\mathbb{C}^{d}$
maximally entangled states is bounded above by $\frac{d}{k}$. For
$\mathcal{B}^{\prime}$, $d=2^{n}$ and $k=2^{2n}$. Hence 
\begin{eqnarray*}
\mathbb{F}_{{\rm local}}\left(\mathcal{B}^{\prime}\right) & \leq & \frac{2^{n}}{2^{2n}}=\frac{1}{2^{n}}
\end{eqnarray*}
On the other hand, using the resource state resource state $\vert\Psi_{m}\rangle$
we can teleport $m$ of our $n$ systems from one party to the other,
allowing us to identify the value of $(i_{1},i_{2},\ldots,i_{m})$.
This reduces the original problem to that of distinguishing $\left(n-m\right)$
unknown Bell states for which the optimal local fidelity is exactly
$\frac{1}{2^{(n-m)}}$ (since the bound from \cite{BN-2013} is saturated
by measuring in the computational basis). Thus, $\mathbb{F}_{{\rm local}}\left(\Psi_{m}\otimes\mathcal{B}^{\prime}\right)=\frac{1}{2^{(n-m)}}>\frac{1}{2^{n}}\ge\mathbb{F}_{{\rm local}}\left(\mathcal{B}^{\prime}\right)$
for $1\leq m\leq n$. Note that $\mathcal{B}^{\prime}$ become perfectly
distinguishable using this protocol if and only if $m=n$. 
\end{example}
The following example shows that not every pure entangled state can
be useful for a fixed set $\mathcal{S}$ of LI states. 
\begin{example}
\label{GHZ-example} Consider the problem of local discrimination
of the three-qubit $GHZ$ basis $\mathcal{G}=\left\{ \vert\Phi_{i}\rangle\right\} $,
$i=1,\dots,8$: 
\[
\begin{array}{cccc}
\vert\Phi_{1}\rangle=\frac{1}{\sqrt{2}}\left(\vert000\rangle+\vert111\rangle\right) &  &  & \vert\Phi_{5}\rangle=\frac{1}{\sqrt{2}}\left(\vert010\rangle+\vert101\rangle\right)\\
\vert\Phi_{2}\rangle=\frac{1}{\sqrt{2}}\left(\vert000\rangle-\vert111\rangle\right) &  &  & \vert\Phi_{6}\rangle=\frac{1}{\sqrt{2}}\left(\vert010\rangle-\vert101\rangle\right)\\
\vert\Phi_{3}\rangle=\frac{1}{\sqrt{2}}\left(\vert001\rangle+\vert110\rangle\right) &  &  & \vert\Phi_{7}\rangle=\frac{1}{\sqrt{2}}\left(\vert011\rangle+\vert100\rangle\right)\\
\vert\Phi_{4}\rangle=\frac{1}{\sqrt{2}}\left(\vert001\rangle-\vert110\rangle\right) &  &  & \vert\Phi_{8}\rangle=\frac{1}{\sqrt{2}}\left(\vert011\rangle-\vert100\rangle\right)
\end{array}
\]
We assume that the states are all equiprobable. Let the qubits be
labeled as $A,$ $B$ and $C$. Let $\mathbb{F}_{{\rm local}}^{\left(i\vert jk\right)}\left(\mathcal{G}\right)$
denote the optimal fidelity across a bipartition $\left(i\vert jk\right)$,
where $i\neq j\neq k\in\left\{ A,B,C\right\} $. We show that local
distinguishability of the above states cannot be enhanced by any bipartite
pure state shared between any two parties.

The proof is by contradiction. Suppose that a bipartite pure state
$\vert\Psi_{jk}\rangle$ shared between two parties $j$ and $k$
can enhance the local distinguishability of the states in $\mathcal{G}$.
This means that there must exist a local protocol ${\bf P}$ such
that the inequalities 
\begin{equation}
\mathbb{F}_{{\rm local}}\left(\mathcal{G}\right)<\mathbb{F}\left(\Psi_{jk}\otimes\mathcal{G}\vert\mathbf{P}\right)\leq\mathbb{F}_{{\rm local}}^{\left(i\vert jk\right)}\left(\mathcal{G}\right)\label{inequality1}
\end{equation}
must hold. Since across every bipartition (for example, $A\vert BC$)
each state has a maximum Schmidt coefficient of $\frac{1}{2}$, we
can apply the result \cite{sep-upper-bound} to obtain 
\begin{eqnarray*}
\mathbb{F}_{{\rm local}}^{\left(i\vert jk\right)}\left(\mathcal{G}\right)\le\mathbb{F}_{{\rm sep}}^{\left(i\vert jk\right)}\left(\mathcal{G}\right) & \leq & \frac{1}{2}:i\neq j\neq k\in\left\{ A,B,C\right\} 
\end{eqnarray*}
As the separable fidelity in a multipartite setting is bounded by
the minimum separable fidelity across all bi-partitions, 
\begin{eqnarray*}
\mathbb{F}_{{\rm sep}}\left(\mathcal{G}\right) & \leq & \min\left\{ \mathbb{F}_{{\rm sep}}^{\left(A\vert BC\right)}\left(\mathcal{G}\right),\mathbb{F}_{{\rm sep}}^{\left(B\vert AC\right)}\left(\mathcal{G}\right),\mathbb{F}_{{\rm sep}}^{\left(C\vert AB\right)}\left(\mathcal{G}\right)\right\} \\
 & \leq & \frac{1}{2}
\end{eqnarray*}
This upper bound is attained by LOCC simply by measuring in the computational
basis and decoding with one of the two possible inputs as in Example
\ref{Lattice State Example}. Thus we have proved that $\mathbb{F}_{{\rm local}}\left(\mathcal{G}\right)=\mathbb{F}_{{\rm local}}^{\left(i\vert jk\right)}\left(\mathcal{G}\right)=\frac{1}{2}$
which is in contradiction with (\ref{inequality1}). 
\end{example}
The above example can be immediately extended to the problem of distinguishing
a $N$-qubit $GHZ$ basis using a $N^{\prime}$-partite resource state;
and a similar argument exploiting the symmetry properties of a $GHZ$
basis shows that a $N^{\prime}$-partite state, no matter how entangled,
cannot be useful as a resource if $N^{\prime}\leq N-1$.

\subsection{Optimal resource states }

We now come to the question: When is a resource state optimal for
local discrimination of a given set $\mathcal{S}$ of locally indistinguishable
states? We have the following definition. 
\begin{defn}
\label{optimal-def} For a set $\mathcal{S}$ of locally indistinguishable
states, let $\mathcal{R}$ be the set of all resource states $\vert\Phi\rangle$
such that 
\begin{eqnarray*}
\mathbb{F}_{{\rm local}}\left(\Phi\otimes\mathcal{S}\right) & = & \mathbb{F}_{{\rm opt}}\left(\mathcal{S}\right)
\end{eqnarray*}
A resource state $\vert\Psi\rangle\in\mathcal{H}_{{\rm \Psi}}$ is
optimal if $\vert\Psi\rangle\in\mathcal{R}$ and if there exists a
well-defined entanglement measure $E$ such that 
\begin{itemize}
\item {} $E\left(\Psi\right)\leq E\left(\Psi^{\prime}\right)$ for any
$\vert\Psi^{\prime}\rangle\in\mathcal{R}$. 
\item {} $\dim\mathcal{H}_{{\rm \Psi}}\leq\dim\mathcal{H}_{{\rm \Psi^{\prime\prime}}}$,
for any $\vert\Psi^{\prime\prime}\rangle\in\mathcal{R}$ such that
$E\left(\Psi^{\prime\prime}\right)=E\left(\Psi\right)$. 
\end{itemize}
\end{defn}
In bipartite systems $E$ can be chosen to be entanglement entropy
\cite{Entanglement-horodecki} which is exactly computable for pure
states and has a clear physical interpretation. For multipartite systems
one may consider a measure that is most suitable to the specific problem
under consideration. For the optimality results presented later in
this paper we have chosen the Schmidt measure \cite{Schmidt-measure}
of entanglement.

For a given set $\mathcal{S}$ of LI states what necessary conditions
must an optimal resource state satisfy? In general, this is difficult
to answer, but some conditions can still be had when $\mathcal{S}$
is an orthonormal basis in $\mathcal{H}_{{\rm s}}=\otimes_{i=1}^{N}\mathbb{C}^{d_{i}}$.
We further suppose that the states$\in\mathcal{S}$ are all equally
likely. Let $\overline{E}\left(\mathcal{S}_{A\vert B}\right)\geq0$
denote the average entanglement of the states a bipartition $A\vert B$
where $E$ is entanglement entropy. Let $\left|\Psi\right\rangle \in\otimes_{i=1}^{N^{\prime}}\mathbb{C}^{d_{i}}$
be a resource state where $2\leq N^{\prime}\leq N$. The following
proposition is an immediate consequence of the entropy bound proved
in \cite{BBKW-2009}: 
\begin{prop}
\label{entanglement-constraint-on-resource} If $\mathbb{F}_{{\rm local}}\left(\Psi\otimes\mathcal{S}\right)=1$,
then $E\left(\Psi_{A\vert B}\right)\geq\overline{E}\left(\mathcal{S}_{A\vert B}\right)\geq0$
for every bipartition $A\vert B$. 
\end{prop}
If $\mathcal{S}$ contains only product states then the above proposition
doesn't yield any useful information. On the other hand, if some or
all states in $\mathcal{S}$ are entangled then the above proposition
can narrow down the search for optimal resources. 
\begin{cor}
\label{N-partite condition} Let $\overline{E}\left(\mathcal{S}_{A\vert B}\right)>0$
for every bipartition $A\vert B$, and $\mathbb{F}_{{\rm local}}\left(\Psi\otimes\mathcal{S}\right)=1$.
Then $\left|\Psi\right\rangle $ must be $N$-partite. 
\end{cor}
The corollary can be immediately applied to the local discrimination
problems involving multipartite orthogonal bases where every basis
state is entangled across all bipartitions, e.g. $GHZ$ and $W$ basis.
In what follows, we give couple of examples to illustrate the idea
of optimal resources. 
\begin{example}
Let $\mathcal{S}$ be a maximally entangled orthonormal basis in $\mathbb{C}^{d}\otimes\mathbb{C}^{d}$,
$d\geq2$. If a resource state $\vert\Phi\rangle$ enables perfect
discrimination of $\mathcal{S}$ by LOCC then from Proposition \ref{entanglement-constraint-on-resource}
we see that $E\left(\Phi\right)\geq\log d\; ebits$ which in turn
implies that $\dim\mathcal{H}_{\Phi}\geq d^{2}$. Since for any maximally
entangled state $\vert\Psi\rangle\in\mathbb{C}^{d}\otimes\mbox{\ensuremath{\mathbb{C}}}^{d}$
we have $\mathbb{F}_{{\rm local}}\left(\Psi\otimes\mathcal{S}\right)=1$,
$E\left(\Psi\right)=\log d\; ebits$ and $\dim\mathcal{H}_{\Psi}=d^{2}$,
$\vert\Psi\rangle$ is therefore optimal.

Using similar arguments it's also easy to see that a maximally entangled
state $\left|\Psi\right\rangle $ in $\mathbb{C}^{2}\otimes\mathbb{C}^{2}$
is optimal for distinguishing the set $\mathcal{S}_{1}=\left\{ \mbox{any three Bell states}\right\} $
as well as the set $\mathcal{S}_{2}=\left\{ \left|\Phi_{1}\right\rangle ,\left|\Phi_{2}\right\rangle ,\left|01\right\rangle ,\left|10\right\rangle \right\} $.
Clearly$\vert\Psi\rangle$ is sufficient; it is also known to be necessary:
see \cite{B-IQC-2015} for $\mathcal{S}_{1}$ and \cite{Yu-Duan-Ying}
for $\mathcal{S}_{2}$. Note that even though $\mathcal{S}_{2}$ is
locally more distinguishable than $\mathcal{S}_{1}$ ($\mathbb{F}_{{\rm local}}\left(\mathcal{S}_{2}\right)=\frac{3}{4}$
while $\mathbb{F}_{{\rm local}}\left(\mathcal{S}_{1}\right)=\frac{2}{3}$
\cite{BN-2013}), both the sets require the same amount of entanglement
for optimal discrimination.
\end{example}
The next example shows that an optimal resource may not be shared
among all the parties in a multipartite system. 
\begin{example}
Consider the problem of locally distinguishing three-qubit $GHZ$
states $\mathcal{G}^{\prime}=\left\{ \vert\Phi_{i}\rangle\right\} $,
$i=1,\dots,4$, where the states are given by 
\[
\begin{array}{cccc}
\vert\Phi_{1}\rangle_{ABC}=\frac{1}{\sqrt{2}}\left(\vert000\rangle_{ABC}+\vert111\rangle_{ABC}\right) &  &  & \vert\Phi_{3}\rangle_{ABC}=\frac{1}{\sqrt{2}}\left(\vert001\rangle_{ABC}+\vert110\rangle_{ABC}\right)\\
\vert\Phi_{2}\rangle_{ABC}=\frac{1}{\sqrt{2}}\left(\vert000\rangle_{ABC}-\vert111\rangle_{ABC}\right) &  &  & \vert\Phi_{4}\rangle_{ABC}=\frac{1}{\sqrt{2}}\left(\vert001\rangle_{ABC}-\vert110\rangle_{ABC}\right)
\end{array}
\]
Assume that the states are equally likely. The above states are locally
indistinguishable because they cannot be perfectly distinguished by
LOCC across the bipartition $C\vert AB$. This follows from the observation
that across the bipartition $C\vert AB$ the states are locally equivalent
to the Bell basis. On the other hand, the states can be perfectly
distinguished by LOCC across the two other bipartitions $A\vert BC$
and $B\vert CA$.

We will show that a two-qubit Bell state $\left|\Psi\right\rangle _{BC}=\frac{1}{\sqrt{2}}\left(\left|00\right\rangle _{BC}+\left|11\right\rangle _{BC}\right)$
shared between $B$ and $C$ is an optimal resource. The LOCC protocol
is as follows. In the first step $A$ performs a measurement in the
$\left|\pm\right\rangle =\frac{1}{\sqrt{2}}\left(\left|0\right\rangle \pm\left|1\right\rangle \right)$
basis on the qubit he/she holds as part of the unknown state and informs
the outcome to $B$. If the outcome is $+$, $B$ does nothing. If
the outcome is $-$, $B$ applies $\sigma_{z}$ on the qubit that
he/she holds as part of the unknown state. This results in an unknown
Bell state between $B$ and $C$ with the following mapping: 
\[
\begin{array}{cccc}
\vert\Phi_{1}\rangle_{ABC}\rightarrow\left|\Phi^{+}\right\rangle _{BC} &  &  & \vert\Phi_{3}\rangle_{ABC}\rightarrow\left|\Psi^{^{+}}\right\rangle _{BC}\\
\vert\Phi_{2}\rangle_{ABC}\rightarrow\left|\Phi^{-}\right\rangle _{BC} &  &  & \vert\Phi_{4}\rangle_{ABC}\rightarrow\left|\Psi^{-}\right\rangle _{BC}
\end{array}
\]
The protocol is completed by $B$ and $C$ who distinguish the unknown
Bell state using $\left|\Psi\right\rangle _{BC}$ following the teleportation
protocol. Note that the necessary conditions in Proposition \ref{entanglement-constraint-on-resource}
and Corollary \ref{N-partite condition} are not violated because
$\mathcal{G}^{\prime}$ is not a complete basis. 
\end{example}

\section{Optimal resources for $GHZ$ and graph states}

\subsection{Discrimination of a $GHZ$ basis by LOCCE }

Let $\mathcal{G}_{m}^{N}$ denote a $N$-qubit $m$-partite $GHZ$
basis for $2\leq m\leq N$. Assuming that the $i^{{\rm th}}$ party
holds $n_{i}$ qubits, $1\leq n_{i}\leq\left(N-m+1\right)$ the basis
states are given by a collection of $2^{N-1}$ conjugate pairs 
\begin{eqnarray}
\left|\Phi_{\alpha}^{\pm}\right\rangle _{{\rm N,m}} & = & \frac{1}{\sqrt{2}}\left(\left|\mathbf{k}_{1}^{\alpha}\right\rangle \left|\mathbf{k}_{2}^{\alpha}\right\rangle \cdots\left|\mathbf{k}_{m}^{\alpha}\right\rangle \pm\left|\overline{\mathbf{k}_{1}^{\alpha}}\right\rangle \left|\overline{\mathbf{k}_{2}^{\alpha}}\right\rangle \cdots\left|\overline{\mathbf{k}_{m}^{\alpha}}\right\rangle \right),\alpha=1,\dots,2^{N-1}\label{(N,m)GHZ-basis-1}
\end{eqnarray}
where for every $j$, $j=1,\dots,m$ ${\bf k}_{j}$ is a $n_{j}$-bit
binary string and $\overline{{\bf k}_{j}}$ is its bit-wise orthogonal
complement and $\mathbf{k}_{1}^{\alpha}\mathbf{k}_{2}^{\alpha}\cdots\mathbf{k}_{m}^{\alpha}\neq\mathbf{k}_{1}^{\beta}\mathbf{k}_{2}^{\beta}\cdots\mathbf{k}_{m}^{\beta}$
whenever $\alpha\neq\beta$. The state space of the states with respect
to this partitioning is given by $\mathcal{H}_{{\rm S}}=\otimes_{i=1}^{m}\mathbb{C}^{2^{n_{i}}}$. 

The set $\mathcal{G}_{m}^{N}$ of states is locally indistinguishable
for any partitioning of the $N$ qubits among $m$ parties. This is
because across any bipartition the states remain entangled (by inspection),
and we know that any bipartite orthonormal basis containing entangled
states is locally indistinguishable \cite{HSSH2003}. 

We wish to obtain an optimal resource for distinguishing the elements
of $\mathcal{G}_{m}^{N}$. Let us suppose that a resource state $\left|\Psi\right\rangle $
is optimal. Since the average entanglement of the states in $\mathcal{G}_{m}^{N}$
is one $ebit$ across every bipartition, Corollary \ref{N-partite condition}
tells us that $\left|\Psi\right\rangle $ must be $m$-partite; that
is, $\mathcal{H}_{{\rm \Psi}}=\otimes_{i=1}^{m}\mathbb{C}^{d_{i}}$,
$d_{i}\geq2$ and therefore, $\dim\mathcal{H}_{\Psi}\geq2^{m}$. We
now state the main result: 
\begin{thm}
\label{GHZ thm} A $m$-qubit $GHZ$ state $\left|\Phi\right\rangle $
is an optimal resource for distinguishing the states in $\mathcal{G}_{m}^{N}$
using LOCC for any $2\leq m\leq N$ and any partitioning of the $N$
qubits among $m$ parties. 
\end{thm}
Without loss of generality let $\left|\Phi\right\rangle =\frac{1}{\sqrt{2}}\left(\vert0\rangle^{\otimes m}+\vert1\rangle^{\otimes m}\right)$.
To establish optimality it suffices to show that $\mathbb{F}_{{\rm local}}\left(\Phi\otimes\mathcal{G}_{m}^{N}\right)=1$.
This is because the other conditions for optimality are satisfied
as follows: First, $E_{{\rm SM}}\left(\Phi\right)=1$, where $E_{{\rm SM}}$
is the Schmidt measure \cite{Schmidt-measure} of entanglement. Since
for any multipartite pure state $\left|\phi\right\rangle $, $E_{{\rm SM}}\left(\phi\right)\geq1$,
entanglement of $\left|\Phi\right\rangle $ thus achieves the minimum.
Next, the dimension of the resource space being $2^{m}$ also achieves
the minimum dimension required by any $m$-partite state. What remains
to show is that, in fact, $\mathbb{F}_{{\rm local}}\left(\Phi\otimes\mathcal{G}_{m}^{N}\right)=1$.
We begin with the case $m=N$: 
\begin{lem}
\label{NN-GHZ-by-NN-GHZ-state} For any $N\geq2$, a $GHZ$ basis
$\mathcal{G}_{N}^{N}$ and a resource state $\left|\Phi^{\prime}\right\rangle =\frac{1}{\sqrt{2}}\left(\vert0\rangle^{\otimes N}+\vert1\rangle^{\otimes N}\right)$,
we have $\mathbb{F}_{{\rm local}}\left(\Phi^{\prime}\otimes\mathcal{G}_{N}^{N}\right)=1$. 
\end{lem}
The proof is given in appendix A. The LOCC protocol, however, is simple.
It consists of a sequence of Bell measurements by each party $A_{i}$,
$i=1,\dots,N$ followed by appropriate Pauli corrections. While any
sequence works, in the proof we assume that the sequence $A_{1}\rightarrow A_{2}\rightarrow\cdots\rightarrow A_{N}$
is followed. The first Bell measurement by $A_{1}$ entangles the
resource state and the shared unknown state but does not eliminate
any state. However, in each subsequent measurement performed in the
order $A_{2}\rightarrow\cdots\rightarrow A_{N-1}$, the outcome maps
exactly half of the candidate states (remaining in that round) onto
a new set of orthonormal states and eliminates the rest. In the end
$A_{N}$ who is left with the task of distinguishing four Bell states
completes the protocol by performing a Bell measurement.

We now extend the theorem to the case $m<N$ with the following lemma
which shows that local discrimination of the states in $\Phi\otimes\mathcal{G}_{m}^{N}$
cannot be harder than the states in $\Phi^{\prime}\otimes\mathcal{G}_{N}^{N}$. 
\begin{lem}
\label{reduction-lemma-version2-1} For any $N\ge m\ge2$, 
\[
\mathbb{F}_{{\rm local}}\left(\Phi\otimes\mathcal{G}_{m}^{N}\right)=\mathbb{F}_{{\rm local}}\left(\Phi_{m}^{\prime}\otimes\mathcal{G}_{m}^{N}\right)\ge\mathbb{F}_{{\rm local}}\left(\Phi^{\prime}\otimes\mathcal{G}_{N}^{N}\right)
\]
where $\left|\Phi_{m}^{\prime}\right\rangle $ is the $N$-qubit $GHZ$
state shared by $m$ parties with the same partitioning of the qubits
as in $\mathcal{G}_{m}^{N}$. 
\end{lem}
The first equality is easy to see, since $\left|\Phi\right\rangle \leftrightarrow\left|\Phi_{m}^{\prime}\right\rangle $
by LOCC with unit probability. If we start with $\vert\Phi\rangle$
then we can locally transform $\vert\Phi\rangle$ into $\vert\Phi_{m}^{\prime}\rangle$
by having each party append $(n_{i}-1)$ additional qubits and then
perform a control-NOT on each to entangle them, mapping $\vert0\rangle$
to $\vert0\rangle^{\otimes n_{i}}$ and $\vert1\rangle$ to $\vert1\rangle^{\otimes n_{i}}$.
Conversely, if we start with $\left|\Phi_{m}^{\prime}\right\rangle $,
each party simply performs the inverse unitary which is again a control-NOT
to disentangle the additional qubits to arrive at $\left|\Phi\right\rangle $.

The proof of the second inequality comes from the observation that
the sets $\Phi_{m}^{\prime}\otimes\mathcal{G}_{m}^{N}$ and $\Phi^{\prime}\otimes\mathcal{G}_{m}^{N}$
contain the same states but in the latter the states are partitioned
further. Thus, the problem of locally distinguishing the elements
of the former cannot be harder than locally distinguishing the elements
of the latter. Hence the lemma.

The proof the theorem immediately follows since $\mathbb{F}_{{\rm opt}}\left(\mathcal{G}_{N}^{N}\right)=\mathbb{F}_{{\rm local}}\left(\Phi^{\prime}\otimes\mathcal{G}_{N}^{N}\right)=1$
by Lemma \ref{NN-GHZ-by-NN-GHZ-state}.

\subsection{Discrimination of Graph states by LOCCE }

Let ${\cal P}_{N}$ be the set of $N$-fold tensor products of Pauli
operators: ${\cal P}_{N}=\{\otimes_{k=1}^{N}\sigma_{k}\}$. For any
graph $G$ on $N$ vertices, we can define an associated set of $N$-partite
graph state in which each party holds a qubit. Following the definitions
in, e.g. \cite{HEB-2004,DAB-2003}, each vertex $a\in V$ is associated
with a unitary $K_{G}^{(a)}=\otimes_{k=1}^{N}\sigma_{k_{i}}\in{\cal P}_{N}$,
where $\sigma_{i_{i}}=X$; $\sigma_{k_{i}}=Z$ if $v_{k}$ and $v_{i}$
are neighbors in $G$; and $\sigma_{k_{i}}={\cal I}$ otherwise. %Since we will be keeping our graph $G$ fixed, we'll write $K_a := K_G^{(a)}$. 

The set of operators $\left\{ K_{G}^{(a)}:a\in V\right\} $ commute
and (except in degenerate cases) define a unique basis of common eigenvectors,
${\cal S}\subset\left(\mathbb{C}^{2}\right)^{\otimes N}$. We call
these the \textit{graph states} corresponding to the graph $G$. We
identify the state $\left|\Psi_{G}\right\rangle \in{\cal S}$ as the
unique state which is simultaneously an eigenvector of each $K_{G}^{(a)}$
with eigenvalue one, and we propose that $\left|\Psi_{G}^{*}\right\rangle $
is an optimal resource to distinguish the elements of ${\cal S}$
(where $^{*}$ denotes the entrywise complex conjugate). We start
with the following elementary observation: 
\begin{lem}
The set of graph states ${\cal S}$ contains the orbit of $\left|\Psi_{G}\right\rangle $
under the action of ${\cal P}_{N}$. 
\end{lem}
The proof is straightforward: For any $\left(i_{1},i_{2},\ldots,i_{N}\right)\in\left\{ 0,1,2,3\right\} ^{N}$,
define 
\begin{eqnarray*}
\left|\Psi\right\rangle  & = & \left(\otimes_{k=1}^{N}\sigma_{i_{k}}\right)\left|\Psi_{G}\right\rangle 
\end{eqnarray*}
Since $K_{G}^{(a)}\in{\cal P}_{N}$, we can use standard commutation
properties to get that 
\begin{eqnarray*}
K_{G}^{(a)}\left(\otimes_{k=1}^{N}\sigma_{i_{k}}\right) & = & \pm\left(\otimes_{k=1}^{N}\sigma_{i_{k}}\right)K_{G}^{(a)}
\end{eqnarray*}
This means that 
\begin{eqnarray*}
K_{G}^{(a)}\left|\Psi\right\rangle  & =K_{G}^{(a)}\left(\otimes_{k=1}^{N}\sigma_{i_{k}}\right)\left|\Psi_{G}\right\rangle  & =\pm\left(\otimes_{k=1}^{N}\sigma_{i_{k}}\right)K_{G}^{(a)}\left|\Psi_{G}\right\rangle \\
 & =\pm\left(\otimes_{k=1}^{N}\sigma_{i_{k}}\right)\left|\Psi_{G}\right\rangle  & =\pm\left|\Psi\right\rangle 
\end{eqnarray*}
Hence, for each $a\in V$, $\vert\Psi\rangle$ is an eigenvector of
$K_{G}^{(a)}$, which implies that $\vert\Psi\rangle\in{\cal S}$.

This allows us to state our result: 
\begin{thm}
\label{thm:graph states} For any graph $G$ on $N$ vertices that
uniquely defines states ${\cal S}$ and $\left|\Psi_{G}\right\rangle $
as above: The state $\left|\Psi_{G}^{*}\right\rangle $ is an optimal
resource to distinguish the elements of ${\cal S}$ under LOCCE. 
\end{thm}
The proof is a consequence of the lemma. Suppose we start with the
state $\left|\Psi_{G}^{*}\right\rangle \otimes\left|\Psi_{x}\right\rangle $
for $\left|\Psi_{x}\right\rangle \in{\cal S}$. Each party can measure
their two-qubit system in the Bell basis. This is equivalent to performing
the global measurement 
\begin{eqnarray*}
M & = & \left\{ \left(\mathcal{I}_{2^{N}}\otimes\sigma_{N}\right)\left|\Phi\right\rangle \left\langle \Phi\right\vert \left(\mathcal{I}_{2^{N}}\otimes\sigma_{N}\right):\sigma_{N}\in\mathcal{P}\right\} 
\end{eqnarray*}
where $\vert\Phi\rangle$ is the canonical maximally entangled state.
(Note that the elements of $M$ project onto the Lattice States discussed
in Example \ref{Lattice State Example}.)

Given the initial state $\left|\Psi_{G}^{*}\right\rangle \otimes\left|\Psi_{x}\right\rangle $,
the probability of getting the outcome $\sigma_{N}$ is given by 
\[
\left|\left\langle \Psi_{G}^{*}\otimes\Psi_{x}\left|\left(\mathcal{I}_{2^{N}}\otimes\sigma_{N}\right)\right|\Phi\right\rangle \right|^{2}=\frac{1}{2^{N}}\left|\left\langle \Psi_{x}\left|\sigma_{N}\right|\Psi_{G}\right\rangle \right|^{2}=\frac{1}{2^{N}}\left|\left\langle \Psi_{x}\vert\Psi_{y}\right\rangle \right|^{2}
\]
for some $\left|\Psi_{y}\right\rangle \in{\cal S}$ by the lemma.
Since the elements of ${\cal S}$ are mutually orthogonal, the probability
of getting this outcome is zero unless $y=x$. It also confirms that
each $\left|\Psi_{y}\right\rangle $ corresponds to $2^{N}$ measurement
outcomes. Since there are $4^{N}$ possible outcomes, this implies
that the orbit of ${\cal P}_{N}$ in fact reaches all $2^{N}$ elements
of ${\cal S}$ and that the set $\left|\Psi_{G}^{*}\right\rangle \otimes{\cal S}$
can be perfectly distinguished with LOCC.

On the other hand, the result in \cite{BHN-2016} asserts that for
any optimal resource $\left|\Psi\right\rangle $, $\left|\Psi^{*}\right\rangle $
must be locally transformable into $\left|\Psi_{G}\right\rangle $,
which implies that for every entanglement measure, $E\left(\Psi\right)\ge E\left(\Psi_{G}\right)=E\left(\Psi_{G}^{*}\right)$
and every local space must have dimension at least two. Hence, $\left|\Psi_{G}^{*}\right\rangle $
is an optimal resource state for ${\cal S}$.

Note that the $(m,m)$-$GHZ$ states are locally equivalent to the
graph states corresponding to the complete graph $K_{m}$ on $m$
vertices; hence Theorem \ref{thm:graph states} is a direct generalization
of Theorem \ref{GHZ thm}.

\section{Optimal resources for one-way LOCCE in bipartite systems}

Optimal resource states can be defined with respect to any restricted
set of measurements. One familiar restriction on LOCC in bipartite
systems is that of one-way communication, in which Alice can communicate
her measurement results to Bob but Bob cannot communicate back to
Alice. We denote the optimal fidelity with respect to this restriction
$\mathbb{F}_{{\rm local-1}}$.

Given a set of bipartite states $\left\{ \left|\psi_{i}\right\rangle \right\} \subset\mathbb{C}^{d}\otimes\mathbb{C}^{d}$,
we can identify each state with a $d\times d$ matrix in the standard
way 
\begin{eqnarray*}
\vert\psi_{i}\rangle & = & \left(\mathcal{I}\otimes M_{i}\right)\vert\Phi\rangle
\end{eqnarray*}
where $\vert\Phi\rangle$ is the standard maximally-entangled state
on $\mathbb{C}^{d}\otimes\mathbb{C}^{d}$. It was noted in \cite{Nathanson-2013,BGK}
that a necessary condition for one-way LOCC discrimination is the
existence of a state $\vert\varphi\rangle$ such that the $\left\langle \varphi\left|M_{i}^{*}M_{j}\right|\varphi\right\rangle =0$
whenever $i\ne j$. We can use this condition to state the following: 
\begin{prop}
\label{one-way} Let ${\cal S}=\left\{ \left|\psi_{i}\right\rangle \right\} $
be a complete orthogonal basis of $\mathbb{C}^{d}\otimes\mathbb{C}^{d}$,
and let $\vert\Phi\rangle\in\mathbb{C}^{d}\otimes\mathbb{C}^{d}$
be the standard maximally-entangled state. If ${\cal S}$ contains
at least one state with full Schmidt rank $=d$, then $\vert\Phi\rangle$
is an optimal resource for the problem of one-way LOCC discrimination. \end{prop}
\begin{proof}
It is clear that $\mathbb{F}_{{\rm local-1}}\left(\Phi\otimes\mathcal{S}\right)=1$,
since we can use one-way LOCC to teleport one half of our states to
the other subsystem. What is less clear is that this is optimal, which
we show next.

Suppose that $\vert\Psi\rangle\in\mathbb{C}^{d}\otimes\mathbb{C}^{d}$
is an optimal resource to distinguish our basis $\mathcal{S}$ with
one-way LOCC. We write 
\begin{eqnarray*}
\vert\Psi\rangle=\left({\cal I}\otimes\Lambda^{1/2}\right)\vert\Phi\rangle\qquad\Lambda=\sum_{i=1}^{d}\lambda_{i}\vert i\rangle\langle i\vert
\end{eqnarray*}
%Because at least one state has  Schmidt rank $d$, a necessary condition is that each $\lambda_i$ is strictly positive. 

If we can distinguish the set $\Psi\otimes{\cal S}$ with one-way
LOCC, then there exist positive constants $\{a_{k}\}$ and states
$\{\vert\varphi_{k}\rangle\}\subset\mathbb{C}^{d^{2}}$ such that
$\sum_{k}a_{k}\vert\varphi_{k}\rangle\langle\varphi_{k}\vert={\cal I}_{d^{2}}$
and $\langle\varphi_{k}\vert(\Lambda\otimes M_{i}^{*}M_{j})\vert\varphi_{k}\rangle=0$
for $i\ne j$ \cite{Nathanson-2013}. If we write $\vert\varphi_{k}\rangle=\left({\mathcal{I}}\otimes R_{k}\right)\vert\Phi\rangle$,
then whenever $i\ne j$, 
\begin{eqnarray}
\langle\varphi_{k}\vert(\Lambda\otimes M_{i}^{*}M_{j})\vert\varphi_{k}\rangle=\frac{1}{d}\mbox{Tr}(R_{k}\Lambda R_{k}^{*})M_{i}^{*}M_{j}=0\label{eqn: orthogonality2}
\end{eqnarray}
Since the elements of $S$ are linearly independent, so are the matrices
$\{M_{i}\}$. By assumption, at least one of the states in ${\cal S}$
(say $\left|\psi_{1}\right\rangle $) has Schmidt rank $d$, which
means that the corresponding matrix $M_{1}$ is invertible. This implies
that the matrices $\{M_{1}^{*}M_{j}\}_{j=2}^{d^{2}}$ are linearly
independent; and since they are all traceless, the orthogonal complement
of $\{M_{1}^{*}M_{j}\}_{j=2}^{d^{2}}$ is simply the multiples of
the identity matrix ${\cal I}$. Setting $i=1$ in (\ref{eqn: orthogonality2}),
we get that for any $k$, $R_{k}\Lambda R_{k}^{*}=t_{k}{\cal I}_{d}$
is a multiple of the identity. This implies that each $R_{k}$ is
full rank and that, in fact, for each $k$, there exists a unitary
$U_{k}$ such that 
\begin{eqnarray*}
R_{k}=\sqrt{t_{k}}U_{k}\Lambda^{-1/2}
\end{eqnarray*}
Since $\vert\varphi_{k}\rangle$ is a normalized pure state, $\mbox{Tr}R_{k}^{*}R_{k}=d$,
which implies that $t_{k}=t$ does not depend on $k$. We can now
rewrite our decomposition of the identity to get 
\begin{eqnarray*}
{\cal I}_{d^{2}} & = & \sum_{k}a_{k}\vert\varphi_{k}\rangle\langle\varphi_{k}\vert\\
 & = & \sum_{k}a_{k}t\left({\mathcal{I}}\otimes U_{k}\Lambda^{-1/2}\right)\vert\Phi\rangle\langle\Phi\vert\left({\mathcal{I}}\otimes\Lambda^{-1/2}U_{k}^{*}\right)\\
 & = & \sum_{k}a_{k}t\left(\Lambda^{-1/2}\otimes U_{k}\right)\vert\Phi\rangle\langle\Phi\vert\left(\Lambda^{-1/2}\otimes U_{k}^{*}\right)\\
\Lambda\otimes{\cal I}_{d} & = & \sum_{k}a_{k}t\left({\mathcal{I}}\otimes U_{k}\right)\vert\Phi\rangle\langle\Phi\vert\left({\mathcal{I}}\otimes U_{k}^{*}\right)
\end{eqnarray*}
In the last line, all of the states on the right side are maximally-entangled,
so if we trace out the second system, we get the maximally mixed state,
which implies that 
\begin{eqnarray*}
\Lambda & = & \frac{t}{d}\left(\sum_{k}a_{k}\right){\cal I}_{d}=td{\cal I}_{d}
\end{eqnarray*}
Since $\mbox{Tr}\Lambda=d$, we get that $\Lambda={\cal I}_{d}$.

Conclusion: If the state $\vert\Psi\rangle=\left({\cal I}\otimes\Lambda^{1/2}\right)\vert\Phi\rangle\in\mathbb{C}^{d}\otimes\mathbb{C}^{d}$
can be used as a resource to locally distinguish a complete basis
of $\mathbb{C}^{d}\otimes\mathbb{C}^{d}$ containing a full-rank state,
then $\Psi$ must be maximally-entangled. 
\end{proof}
We note that this same result can be shown using the operator system
methods in the recent work of Kribs, et al. \cite{KMNR-2017}.

\section{Conclusion and open problems }

The notion of entanglement as a resource stems from the fact that
shared entanglement can help us to realize nonlocal quantum operations
on composite systems by LOCC. In this paper, we have considered the
task of quantum state discrimination within the framework of LOCCE,
short of Local Operations, Classical Communication and Entanglement.
To better understand the role of entanglement as a resource, we focused
on the characterization of resource states and defined useful and
optimal resource states for any given local state discrimination problem.
These definitions were further illustrated with results and examples
in both bipartite and multipartite systems.

Some interesting questions emerge from the notion of useful resources.
For example, let $\mathcal{S}$ be a set of LI states in $\mathbb{C}^{d_{1}}\otimes\mathbb{C}^{d_{2}}$,
$3\leq d_{1}\leq d_{2}$. Is there a pure entangled state $\left|\Phi\right\rangle $
of Schmidt rank $r$ such that $r<d_{1}$ which is useful for distinguishing
$\mathcal{S}$? From the example given in this paper we know that
such a pair $\left(\mathcal{S},\Phi\right)$ can be found but a general
answer is wanting. More generally, for a fixed set $\mathcal{S}$,
how can we characterize the set of states $\left|\Phi\right\rangle $
such that $\left|\Phi\right\rangle $ is useful for distinguishing
$\mathcal{S}$? 

Some other open problems which may also be of interest are discussed
below.

Consider, for example, the following orthonormal basis in $\mathbb{C}^{2}\otimes\mathbb{C}^{2}$:
\begin{eqnarray*}
\left|\psi_{1}\right\rangle =\alpha\left|00\right\rangle +\beta\left|11\right\rangle  &  & \left|\psi_{2}\right\rangle =\beta\left|00\right\rangle -\alpha\left|11\right\rangle \\
\left|\psi_{3}\right\rangle =\gamma\left|01\right\rangle +\delta\left|10\right\rangle  &  & \left|\psi_{4}\right\rangle =\delta\left|01\right\rangle -\gamma\left|10\right\rangle 
\end{eqnarray*}
where $\alpha,\beta,\gamma,\delta$ with $\alpha\geq\beta\geq0$ and
$\gamma\geq\delta\geq0$ are real numbers satisfying $\alpha^{2}+\beta^{2}=1$
and $\gamma^{2}+\delta^{2}=1$. The local fidelity of the above set
of states can be shown to be $\mathbb{F}_{{\rm local}}\left(\mathcal{S}\right)=\frac{1}{2}\left(\alpha^{2}+\gamma^{2}\right)$.
The states are locally indistinguishable except for the case $\alpha=\gamma=1$.
i.e. when the set reduces to the computational basis. Clearly, the
states can be perfectly distinguished using a Bell state as resource;
and if $\max(\alpha,\gamma)=1$, this is necessarily optimal \cite{BHN-2016}.
However, we do not know whether this is optimal in other cases, and
it would be useful to understand how this depends on $\alpha$ and
$\gamma$.

Another problem worth considering is motivated by the no-go results
on local distinguishability of maximally entangled states \cite{Ghosh-2001,Ghosh-2004,fan-2005,Nathanson-2005,Yu-Duan-2012,Cosentino-2013,Cosentino-Russo-2014}.
Suppose $\mathcal{S}$ is a set of orthonormal maximally entangled
states in $\mathbb{C}^{d}\otimes\mathbb{C}^{d}$. If the states form
a basis then we know that any maximally entangled state in $\mathbb{C}^{d}\otimes\mathbb{C}^{d}$
is an optimal resource. On the other hand, if the states do not form
a basis, that is, $\left|\mathcal{S}\right|<d^{2}$, they can still
be locally indistinguishable \cite{Ghosh-2001,Nathanson-2005,Yu-Duan-2012,Cosentino-2013,Cosentino-Russo-2014}
and in these cases, except when $d=2$, we do not know the optimal
resources. 

In multipartite systems, questions related to optimal resources may
pose different kinds of challenges, especially because of the complex
structure of the states, computability of entanglement measures and
existence of multiple SLOCC equivalence classes \cite{DVC,GW2011}.
In fact, the existence of multiple SLOCC classes led to a recent no-go
result \cite{BHN-2016} which states that for a given multipartite
system, a universal resource (a state which can optimally distinguish
any set of locally indistinguishable states) almost always does not
exist in the same state space. For example, one cannot find a three-qubit
pure entangled state that can perfectly distinguish any three-qubit
orthonormal basis by LOCC. This in turn implies that any universal
resource for a three-qubit system must belong to higher dimensions.
In view of this, finding optimal resources in multipartite systems
could be challenging. In this paper, we were able to make partial
progress by solving for $GHZ$ and Graph states; however, optimal
resources for distinguishing any other orthonormal basis with states
chosen from other SLOCC classes are not yet known. 
\begin{acknowledgments}
SB is supported in part by SERB project EMR/2015/002373. SH is supported
by fellowships from CSIR, Govt. of India and Bose Institute. \end{acknowledgments}

\section*{Appendix}

\subsection{Proof of Lemma \ref{NN-GHZ-by-NN-GHZ-state}}

A $N$-qubit $N$-partite $GHZ$ basis $\mathcal{G}_{N}^{N}$ is defined
by a collection of $2^{N-1}$ mutually orthogonal conjugate pairs
which can be written as: 
\begin{eqnarray}
\left|\Phi_{\alpha}^{\pm}\right\rangle  & = & \frac{1}{\sqrt{2}}\left(\left|0_{1}^{\alpha}\right\rangle \left|k_{2}^{\alpha}\right\rangle \cdots\left|k_{N}^{\alpha}\right\rangle \pm\left|1_{1}^{\alpha}\right\rangle \left|\overline{k_{2}^{\alpha}}\right\rangle \cdots\left|\overline{k_{N}^{\alpha}}\right\rangle \right),\alpha=1,\dots,2^{N-1}\label{(N,m)GHZ-basis}
\end{eqnarray}
where for every $i=2,\dots,N$, $k_{i}\in\left\{ 0,1\right\} $ and
$\overline{k_{i}}$ is its complement. We now give a LOCC protocol
that perfectly distinguishes the states in $\mathcal{G}_{N}^{N}$
using the resource $\left|\Phi^{\prime}\right\rangle =\frac{1}{\sqrt{2}}\left(\left|0\right\rangle ^{\otimes N}+\left|1\right\rangle ^{\otimes N}\right)$.
First, we write the (unnormalized) states in $\Phi^{\prime}\otimes\mathcal{G}_{N}^{N}$
as 
\begin{eqnarray}
\left|\Phi^{\prime}\right\rangle \otimes\left|\Phi_{\alpha}^{\pm}\right\rangle  & = & \left|0_{r}0_{1}\right\rangle \left|0_{r}k_{2}^{\alpha}\right\rangle \cdots\left|0_{r}k_{N}^{\alpha}\right\rangle \pm\left|1_{r}1_{1}\right\rangle \left|1_{r}\overline{k_{2}^{\alpha}}\right\rangle \cdots\left|1_{r}\overline{k_{N}^{\alpha}}\right\rangle \nonumber \\
 &  & \pm\left|0_{r}1_{1}\right\rangle \left|0_{r}\overline{k_{2}^{\alpha}}\right\rangle \cdots\left|0_{r}\overline{k_{N}^{\alpha}}\right\rangle +\left|1_{r}0_{1}\right\rangle \left|1_{r}k_{2}^{\alpha}\right\rangle \cdots\left|1_{r}k_{N}^{\alpha}\right\rangle ;\;\alpha=1,\dots,2^{N-1}\label{full-set-mm}
\end{eqnarray}
where the subscript \textquotedbl{}$r$\textquotedbl{} indicates that
the qubit belongs to the resource state. The protocol constitutes
a series of sequential Bell measurements by all the parties $A_{i}$,
$i=1,\dots,N$. We adopt the following sequence: $A_{1}\rightarrow A_{2}\rightarrow\cdots\rightarrow A_{N}$.

1. $A_{1}$ performs a Bell measurement on the two qubits and informs
the outcome to $A_{2}$ who applies the appropriate Pauli correction
following the convention of standard teleportation on the resource
qubit he/she holds. This measurement completely disentangles the first
two qubits held by $A_{1}$ and results in a state shared between
the rest of the parties. This resulting state belongs to one of the
two sets $\Phi$ and $\Psi$ (given below) depending on whether the
outcome was in $\left\{ \Phi^{+}/\Phi^{-}\right\} $ or $\left\{ \Psi^{+}/\Psi^{-}\right\} $:
\begin{eqnarray*}
\Phi & : & \left\{ \left|0_{r}k_{2}^{\alpha}\right\rangle \cdots\left|0_{r}k_{N}^{\alpha}\right\rangle \pm\left|1_{r}\overline{k_{2}^{\alpha}}\right\rangle \cdots\left|1_{r}\overline{k_{N}^{\alpha}}\right\rangle \right\} ;\alpha=1,\dots,2^{N-1}\\
\Psi & : & \left\{ \left|0_{r}\overline{k_{2}^{\alpha}}\right\rangle \cdots\left|0_{r}\overline{k_{N}^{\alpha}}\right\rangle \pm\left|1_{r}k_{2}^{\alpha}\right\rangle \cdots\left|1_{r}k_{N}^{\alpha}\right\rangle \right\} ;\;\alpha=1,\dots,2^{N-1}
\end{eqnarray*}
Note that, as of now, the measurement by $A_{1}$ does not eliminate
any state; instead, it entangles the resource state and the unknown
state.

2. Let us suppose that the outcome of the measurement by $A_{1}$
was either $\Phi^{+}$ or $\Phi^{-}$. The resulting state, now shared
between the parties $A_{2},A_{3},\dots,A_{N}$, therefore, belongs
to the set $\Phi$. The task is now to distinguish the elements in
$\Phi$. The states in $\Phi$ can be grouped into two disjoint subsets
$\Phi_{0}$ and $\Phi_{1}$ depending on whether $k_{2}^{\alpha}$
takes the value $0$ or $1$. By an appropriate relabeling of the
states, the sets $\Phi_{0}$ and $\Phi_{1}$ are given by: 
\begin{eqnarray*}
\Phi_{0} & : & \left\{ \left|0_{r}0_{2}\right\rangle \left|0_{r}k_{3}^{\alpha}\right\rangle \cdots\left|0_{r}k_{N}^{\alpha}\right\rangle \pm\left|1_{r}1_{2}\right\rangle \left|1_{r}\overline{k_{3}^{\alpha}}\right\rangle \cdots\left|1_{r}\overline{k_{N}^{\alpha}}\right\rangle \right\} ;\alpha=1,\dots,2^{N-2}\\
\Phi_{1} & : & \left\{ \left|0_{r}1_{2}\right\rangle \left|0_{r}k_{3}^{\alpha}\right\rangle \cdots\left|0_{r}k_{N}^{\alpha}\right\rangle \pm\left|1_{r}0_{2}\right\rangle \left|1_{r}\overline{k_{3}^{\alpha}}\right\rangle \cdots\left|1_{r}\overline{k_{N}^{\alpha}}\right\rangle \right\} ;\;\alpha=2^{N-2}+1,\dots,2^{N-1}
\end{eqnarray*}
Each of the sets $\Phi_{0}$ and $\Phi_{1}$ contains exactly $2^{N-2}$
conjugate pairs. $A_{2}$ now performs a Bell measurement on the two
qubits he/she holds, and informs the result to $A_{3}$ who applies
the appropriate Pauli correction on the resource qubit. This measurement
disentangles the two qubits held by $A_{2}$ and results in a $\left(N-2\right)$-partite
state shared between $A_{3},A_{4},\dots,A_{N}$. The resulting state
belong to one of the following two sets $\Phi^{\prime}$ and $\Psi^{\prime}$
depending on whether the outcome was $\Phi^{+}/\Phi^{-}$ or $\Psi^{+}/\Psi^{-}$:
\begin{eqnarray*}
\Phi^{\prime} & : & \left\{ \left|0_{r}k_{3}^{\alpha}\right\rangle \cdots\left|0_{r}k_{N}^{\alpha}\right\rangle \pm\left|1_{r}\overline{k_{3}^{\alpha}}\right\rangle \cdots\left|1_{r}\overline{k_{N}^{\alpha}}\right\rangle \right\} ;\alpha=1,\dots,2^{N-2}\\
\Psi^{\prime} & : & \left\{ \left|0_{r}k_{3}^{\alpha}\right\rangle \cdots\left|0_{r}k_{N}^{\alpha}\right\rangle \pm\left|1_{r}\overline{k_{3}^{\alpha}}\right\rangle \cdots\left|1_{r}\overline{k_{N}^{\alpha}}\right\rangle \right\} ;\;\alpha=1,\dots,2^{N-2}
\end{eqnarray*}
As the resulting state belong to either $\Phi^{\prime}$ or $\Psi^{\prime}$,
thus, this measurement eliminates $2^{(N-1)}$ states. One can do
a similar analysis had the outcome of $A_{1}$'s measurement was either
$\Psi^{+}$ or $\Psi^{-}$.

3. The protocol continues in a similar fashion, each round eliminating
exactly half of the states that remained to be distinguished in the
previous round; that is, the second round eliminates $2^{(N-1)}$
states (or equivalently $2^{(N-2)}$ conjugate pairs), the third round
eliminates $2^{(N-2)}$ states (or $2^{(N-3)}$ conjugate pairs) and
so on. It is easy to check that after $(N-1)$ rounds of measurements
(note that state elimination starts only from the second round starting
with the measurement by $A_{2}$), all but four states (or two conjugate
pairs) get eliminated. The last party $A_{N}$ therefore performs
a complete orthogonal measurement to distinguish these four states.
This completes the protocol. 

\begin{thebibliography}{10}
\bibitem{Chitambar-LOCC} E. Chitambar, D. Leung, L. Mancinska, M.
Ozols, A. Winter, \textquotedbl{}Everything You Always Wanted to Know
About LOCC (But Were Afraid to Ask)\textquotedbl{}, Commun. Math.
Phys. \textbf{328}, no. 1, pp. 303-326 (2014).

\bibitem{Entanglement-horodecki} R. Horodecki, P. Horodecki, M. Horodecki,
and K. Horodecki, Rev. Mod. Phys. \textbf{81}, 865 (2009).

\bibitem{Teleportation} C. H. Bennett, G. Brassard, C. Crépeau, R.
Jozsa, A. Peres, and W. K. Wootters, ``Teleporting an unknown quantum
state via dual classical and Einstein-Podolsky-Rosen channels,''
Phys. Rev. Lett. \textbf{70}, 1895 (1993).

\bibitem{Densecoding} C. H. Bennett and S. J. Wiesner, ``Communication
via one- and two-particle operators on Einstein-Podolsky-Rosen states,''
Phys. Rev. Lett. \textbf{69}, 2881 (1992).

\bibitem{Jonathan-Plenio-1999} D. Jonathan and M. B. Plenio, ``Entanglement-assisted
local manipulation of pure quantum states'', Phys. Rev. Lett. \textbf{83},
3566 (1999).

\bibitem{Berry-2007} D. W. Berry, ``Implementation of multipartite
unitary operations with limited resources\textquotedblright , Phys.
Rev. A \textbf{75}, 032349 (2007).

\bibitem{Cirac-et-al-2001} J. I. Cirac, W. Dür, B. Kraus, and M.
Lewenstein, ``Entangling operations and their implementation using
a small amount of entanglement\textquotedblright , Phys. Rev. Lett.
\textbf{86}, 544 (2001).

\bibitem{Collins-et-al-2001} D. Collins, N. Linden, S. Popescu, ``Nonlocal
content of quantum operations\textquotedblright , Phys. Rev. A \textbf{64},
032302 (2001).

\bibitem{BBKW-2009} S. Bandyopadhyay, G. Brassard, S. Kimmel, W.
K. Wootters, ``Entanglement Cost of Nonlocal Measurements'', Phys.
Rev. A \textbf{80}, 012313 (2009).

\bibitem{BRW-2010} S. Bandyopadhyay, R. Rahaman, W. K. Wootters,
``Entanglement cost of two-qubit orthogonal measurements'', J. Phys.
A: Math. Theor. \textbf{43}, 455303 (2010).

\bibitem{Bennett-I-99} C. H. Bennett, D. P. DiVincenzo, C. A. Fuchs,
T. Mor, E. Rains, P. W. Shor, J. A. Smolin, and W. K. Wootters, ``Quantum
Nonlocality without Entanglement,'' Phys. Rev. A \textbf{59}, 1070
(1999).

\bibitem{Cohen-2008} S. M. Cohen, ``Understanding entanglement as
resource: Locallly distinguishing unextendible product bases\textquotedblright ,
Phys. Rev. A \textbf{77}, 012304 (2008).

\bibitem{B-IQC-2015} S. Bandyopadhyay, A. Cosentino, N. Johnston,
V. Russo, J. Watrous, and N. Yu, ``Limitations on separable measurements
by convex optimization,\textquotedblright{} IEEE Transactions on Information
Theory, Vol \textbf{61}, Issue 6, Pages: 3593-3604 (2015).

\bibitem{Yu-Duan-Ying} N. Yu, R. Duan, M. Ying, ``Distinguishability
of Quantum States by Positive Operator-Valued Measures with Positive
Partial Transpose'', IEEE Trans. Inform. Theory, vol.\textbf{60},
no.4, pp. 2069-2079, Apr.(2014).

\bibitem{Nathanson-2013} M. Nathanson, ``Three maximally entangled
states can require two-way local operations and classical communication
for local discrimination,'' Phys. Rev. A \textbf{88}, 062316 (2013).

\bibitem{BHN-2016} S. Bandyopadhyay, S. Halder and M. Nathanson,
``Entanglement as a resource for local state discrimination in multipartite
systems'', Phys. Rev. A $\mathbf{94}$, 022311 (2016)

\bibitem{Peres-Wootters-1991} A. Peres and W. K. Wootters, ``Optimal
Detection of Quantum Information,'' Phys. Rev. Lett. \textbf{66},
1119 (1991); S. Massar and S. Popescu, ``Optimal extraction of information
from finite quantum ensembles,'' Phys. Rev. Lett. \textbf{74}, 1259
(1995).

\bibitem{Walgate-2000} J. Walgate, A. J. Short, L. Hardy, and V.
Vedral, ``Local distinguishability of multipartite orthogonal quantum
states,'' Phys. Rev. Lett. \textbf{85}, 4972 (2000).

\bibitem{Virmani-2001} S. Virmani, M. F. Sacchi, M. B. Plenio, D.
Markham, ``Optimal local discrimination of two multipartite pure
states,'' Phys. Lett. A. \textbf{288} (2001).

\bibitem{Ghosh-2001} S. Ghosh, G. Kar, A. Roy, A. Sen (De), and U.
Sen, ``Distinguishability of Bell states,'' Phys. Rev. Lett. \textbf{87},
277902 (2001).

\bibitem{Walgate-2002} J. Walgate and L. Hardy, ``Nonlocality, Asymmetry
and Distinguishing Bipartite States'', Phys. Rev. Lett. \textbf{89},
147901 (2002).

\bibitem{Ghosh-2002} S. Ghosh, G. Kar, A. Roy, D. Sarkar, A. Sen
(De), and U. Sen, ``Local indistinguishability of orthogonal pure
states by using a bound on distillable entanglement,'' Phys. Rev.
A \textbf{65}, 062307 (2002).

\bibitem{Bennett-II-99 +Divin-2003} C. H. Bennett, D. P. DiVincenzo,
T. Mor, P. W. Shor, J. A. Smolin, and B. M. Terhal, ``Unextendible
Product Bases and Bound Entanglement,'' Phys. Rev. Lett. \textbf{82},
5385 (1999); D. P. DiVincenzo, T. Mor, P. W. Shor, J. A. Smolin, B.
M. Terhal, ``Unextendible Product Bases, Uncompletable Product Bases
and Bound Entanglement,'' Comm. Math. Phys. \textbf{238}, 379 (2003).

\bibitem{HSSH2003} M. Horodecki, A. Sen(De), U. Sen, and K. Horodecki,
``Local indistinguishability: more nonlocality with less entanglement,''
Phys. Rev. Lett. \textbf{90}, 047902 (2003).

\bibitem{Ghosh-2004} S. Ghosh, G. Kar, A. Roy, and D. Sarkar, ``Distinguishability
of maximally entangled states'', Phys. Rev. A \textbf{70}, 022304
(2004). 

\bibitem{Nathanson-2005} M. Nathanson, ``Distinguishing bipartite
orthogonal states by LOCC: best and worst cases,'' Journal of Mathematical
Physics \textbf{46}, 062103 (2005).

\bibitem{Watrous-2005} J. Watrous, ``Bipartite subspaces having
no bases distinguishable by local operations and classical communication,''
Phys. Rev. Lett. \textbf{95}, 080505 (2005).

\bibitem{Duan2007} R. Y. Duan, Y. Feng, Z. F. Ji, and M. S. Ying,
``Distinguishing arbitrary multipartite basis unambiguously using
local operations and classical communication,'' Phys. Rev. Lett.
\textbf{98}, 230502 (2007).

\bibitem{Duan-2009} R. Y. Duan, Y. Feng, Y. Xin, and M. S. Ying,
``Distinguishability of quantum states by separable operations,''
IEEE Trans. Inform. Theory \textbf{55}, 1320 (2009).

\bibitem{Calsamiglia-2010} J. Calsamiglia, J. I. de Vicente, R. Munoz-Tapia,
E. Bagan, ``Local discrimination of mixed states,'' Phys. Rev. Lett.
\textbf{105}, 080504 (2010).

\bibitem{BGK} S. Bandyopadhyay, S. Ghosh and G. Kar, ``LOCC distinguishability
of unilaterally transformable quantum states,'' New J. Phys. \textbf{13},
123013 (2011).

\bibitem{Bandyo-2011} S. Bandyopadhyay, ``More nonlocality with
less purity,'' Phys. Rev. Lett. \textbf{106}, 210402 (2011).

\bibitem{Yu-Duan-2012} N. Yu, R. Duan, and M. Ying, ``Four Locally
Indistinguishable Ququad-Ququad Orthogonal Maximally Entangled States,''
Phys. Rev. Lett. \textbf{109}, 020506 (2012).

\bibitem{BN-2013} S. Bandyopadhyay, M. Nathanson, ``Tight bounds
on the distinguishability of quantum states under separable measurements'',
Phys. Rev. A ${\bf 88}$, 052313 (2013).

\bibitem{Cosentino-2013} A. Cosentino, \textquotedbl{}Positive-partial-transpose-indistinguishable
states via semidefinite programming\textquotedbl{}, Phys. Rev. A \textbf{87}
(1), 012321 32 (2013).

\bibitem{Cosentino-Russo-2014} A. Cosentino and V. Russo, \textquotedbl{}Small
sets of locally indistinguishable orthogonal maximally entangled states\textquotedbl{},
Quantum Information \& Computation \textbf{14} (13-14), 1098-1106

\bibitem{Ghosh-accessible information} S. Ghosh, P. Joag, G. Kar,
S. Kunkri, and A. Roy, ``Locally accessible information and distillation
of entanglement'', Phys. Rev. A \textbf{71}, 012321 (2005).

\bibitem{Horodecki-accessible-information} M. Horodecki, J. Oppenheim,
A. Sen De, U. Sen, ``Distillation protocols: Output entanglement
and local mutual information'', Phys. Rev. Lett. \textbf{93}, 170503
(2004).

\bibitem{Terhal2001}B. M. Terhal, D. P. DiVincenzo, and D. W. Leung,
``Hiding bits in Bell states,'' Phys. Rev. Lett. \textbf{86}, 5807
(2001).

\bibitem{DiVincenzo2002} D. P. DiVincenzo, D. W. Leung, and B. M.
Terhal, Quantum data hiding, IEEE Trans. Inf. Theory \textbf{48},
580 (2002).

\bibitem{Eggeling2002} T. Eggeling, and R. F. Werner, ``Hiding classical
data in multipartite quantum states,'' Phys. Rev. Lett. \textbf{89},
097905 (2002).

\bibitem{MatthewsWehnerWinter} W. Matthews, S. Wehner, A. Winter,
``Distinguishability of quantum states under restricted families
of measurements with an application to quantum data hiding,'' Comm.
Math. Phys. \textbf{291}, Number 3 (2009).

\bibitem{Markham-Sanders-2008} D. Markham and B. C. Sanders, ``Graph
states for quantum secret sharing,'' Phys. Rev. A \textbf{78}, 042309
(2008).

\bibitem{Navascues} M. Navascués, ``Pure state estimation and the
characterization of entanglement,'' Phys. Rev. Lett. \textbf{100},
070503 (2008).

\bibitem{Fuchs-Sasaki} C. A. Fuchs, and M. Sasaki, ``Squeezing quantum
information through a classical channel: measuring the `quantumness'
of a set of quantum states,'' Quantum Information \& Computation
\textbf{3}, 377 (2003).

\bibitem{Vidal-1999} G. Vidal, ``Entanglement of pure states for
a single copy,'' Phys. Rev. Lett. \textbf{83}, 1046 (1999).

\bibitem{sep-upper-bound} For a set of $N$ equally likely orthogonal
states $\vert\psi_{1}\rangle,\vert\psi_{2}\rangle,\dots,\vert\psi_{N}\rangle$
in $\mathbb{C}^{d_{1}}\otimes\mathbb{C}^{d_{2}}$, we have 
\begin{eqnarray}
\mathbb{F}_{{\rm sep}} & \leq & \frac{\lambda_{\max}d_{1}d_{2}}{N}\label{cor-2-eq-1}
\end{eqnarray}
\emph{ where $\lambda_{\max}=\max_{i}\lambda_{i}$, $\sqrt{\lambda_{i}}$
being the largest Schmidt coefficient of the state $\vert\psi_{i}\rangle$
\cite{BN-2013}.}

\bibitem{Schmidt-measure} Any $N$-partite pure state $\vert\phi\rangle\in\otimes_{i=1}^{N\geq3}\mathbb{C}^{d_{i}}$
can be written in the form 
\begin{eqnarray}
\vert\phi\rangle & = & \sum_{i=1}^{R}\alpha_{i}\vert\phi_{1}^{(i)}\rangle_{1}\otimes\vert\phi_{2}^{(i)}\rangle\otimes\cdots\otimes\vert\phi_{N}^{(i)}\rangle,\label{phi}
\end{eqnarray}
where $\vert\phi_{j}^{(i)}\rangle\in\mathbb{C}^{d_{j}}$, $j=1,\dots,N$,
$\alpha_{i}\in\mathbb{C}$, $i=1,\dots,R$ for some $R$. Suppose
that $r$ is the minimal number of product terms $R$ in such a decomposition.
The Schmidt measure \cite{EB2001} is defined as $E\left(\phi\right)=\log_{2}r$.

\bibitem{EB2001} J. Eisert and H.-J. Briegel, ``The Schmidt Measure
as a Tool for Quantifying Multi-Particle Entanglement,'' \url{http://arxiv.org/abs/quant-ph/0007081v3}

\bibitem{HEB-2004} M. Hein, J. Eisert, H.J. Briegel, ``Multi-party
entanglement in graph states,'' Phys. Rev. A \textbf{69}, 062311
(2004).

\bibitem{DAB-2003} Wolfgang Dür, Hans Aschauer, and H-J. Briegel,
``Multiparticle entanglement purification for graph states,\textquotedbl{}
Phys. Rev. Lett. \textbf{91}, 107903(2003).

\bibitem{KMNR-2017} D. Kribs, C. Mintah, M. Nathanson, and R. Pereira,
``Operator structures and quantum one-way locc conditions,'' Journal
of Mathematical Physics \textbf{58}, 092201 (2017).

\bibitem{fan-2005} H. Fan, ``Distinguishability and indistinguishability
by local operations and classical communication'', Phys. Rev. Lett.
\textbf{92}, 177905 (2004).

\bibitem{DVC} W. Dur, G. Vidal, J. I. Cirac, ``Three qubits can
be entangled in two inequivalent ways'', Phys. Rev. A \textbf{62},
062314 (2000).

\bibitem{GW2011} G. Gour and N. R. Wallach, ``Necessary and sufficient
conditions for local manipulation of multipartite pure quantum states,''
New J. Phys. \textbf{13} 073013 (2011).

\end{thebibliography}
\end{document}